\documentclass[draft]{llncs}
\usepackage{cite}
\usepackage{etoolbox}
\usepackage{ezutils}
\usepackage{ezcode}
\usepackage{stmaryrd}
\usepackage{relsize}
\usepackage{xcolor}

\newrobustcmd\Cpp{%
  C\nolinebreak[4]\hspace{-.05em}\raisebox{.4ex}{\relsize{-3}{\bfseries ++}}%
}

\newcommand{\pto}{\rightharpoonup}

\newcommand{\Sig}{\mathcal{L}}

\newcommand{\feq}{\mathop{\simeq}}
\newcommand{\fneq}{\mathop{\not\simeq}}
\newcommand{\f}{\mathsf}
\newcommand{\fnext}{\f{next}}
\newcommand{\flseg}{\f{lseg}}

\newcommand{\ftrue}{\top}
\newcommand{\ffalse}{\bot}
\newcommand{\fnil}{\f{nil}}

\newcommand{\fvalid}{\f{valid}}
\newcommand{\finvalid}{\f{invalid}}
\newcommand{\femp}{\f{emp}}

\newcommand{\SmallFoot}{\textsf{SmallFoot}}
\newcommand{\Separation}{\textsf{slp}}

\newcommand{\SlpMT}{\textsf{Aster{\ooalign{\hidewidth\clap{\raise0.8ex\hbox{\tiny*}}\hidewidth\cr\i}}x}}

\DeclareMathOperator{\dom}{dom}

\newcommand{\Int}{\textsf{Int}}
\newcommand{\Bool}{\textsf{Bool}}

\newcommand{\Val}{\textsf{Val}}
\newcommand{\Var}{\textsf{Var}}

\newcommand\DeclareMathFunction[2]{\newcommand#1{\text{\normalfont\itshape #2}}}
\DeclareMathFunction{\Collide}{collide}
\DeclareMathFunction{\Alloc}{alloc}
\DeclareMathFunction{\Empty}{empty}
\DeclareMathFunction{\WellFormed}{well-formed}
\DeclareMathFunction{\Addr}{addr}

\DeclareMathFunction{\UnfGuard}{match-step} 
\DeclareMathFunction{\UnfCheck}{enclosed} 
\DeclareMathFunction{\UnfUpdate}{residue}

\DeclareMathFunction{\Match}{match}
\DeclareMathFunction{\Prove}{prove}
\DeclareMathFunction{\UnfoldCheck}{\UnfGuard}
\DeclareMathFunction{\UnfoldStep}{\UnfUpdate}

\newcommand\SigmaI{\hat\Sigma}

\DeclareSequence{\sepof}{\ast}
\DeclareSequence{\prodof}{\times}
\DeclareSequence{\disjunionof}{\uplus}
\DeclareSequence{\landof}{\land}

\newcommand\comp{\circ}

\newcommand\dt{.\,}

\title{Separation Logic Modulo Theories}
\author{Juan Antonio Navarro P\'{e}rez\inst{1} \and Andrey Rybalchenko\inst{2}}
\institute{University College London \and Technische Universit\"{a}t M\"{u}nchen}

\begin{document}
\maketitle

\begin{abstract}
Logical reasoning about program data often requires dealing with heap structures as well as scalar data types.
Recent advances in Satisfiability Modular Theory (SMT) already offer efficient procedures for dealing with scalars, yet they lack any support for dealing with heap structures.
In this paper, we present an approach that integrates Separation Logic---a prominent logic for reasoning about list segments on the heap---and SMT.
We follow a model-based approach that communicates aliasing among heap cells between the SMT solver and the Separation Logic reasoning part.
An experimental evaluation using the Z3 solver indicates that our
approach can effectively put to work the advances in SMT for dealing
with heap structures.
This is the first decision procedure for the combination of separation
logic with SMT theories.

\end{abstract}

{\section{Introduction}

Satisfiability Modulo Theory (SMT) solvers play an important role for
the construction of abstract interpretation
tools~\cite{CC77,CousotSMT11}.
They can efficiently deal with relevant logical theories of various
scalar data types, e.g., fixed length bit-vectors and numbers, as well
as uninterpreted functions and
arrays~\cite{Simplify,Z3,Yices,CVC3,MathSAT}.
However, dealing with programs that manipulate heap-allocated data
structures using pointers exposes limitations of today's SMT solvers.

For example, SMT does not support separation logic---a promising
logic for dealing with programs that manipulate the heap following a
certain discipline~\cite{Rey02}.
Advances in the construction of such a solver could directly boost a
wide range of separation logic based verifiers: manual/tool assisted
proof development~\cite{SchorrWaite,YNOT,Appel09}, extended static
checking~\cite{DisPar08,BerCal06}, and automatic inference of heap
shapes~\cite{SpaceInvader,SlayerCAV07,SlayerCAV08,CalDis09}.

In this paper we present a method for extending an STM solver with
separation logic with list segment predicate~\cite{BerCal04}, which is
a frequently used instance of separation logic used by the majority of
existing tools.
Our method decides entailments of the form $\Pi \land \Sigma
\rightarrow \Pi' \land \Sigma'$.
Here, $\Pi$ and $\Pi'$ are \emph{arbitrary} theory assertions
supported by SMT, while $\Sigma$ and $\Sigma'$ are
spatial conjunctions of pointer predicates $\fnext(x,y)$ and list
segment predicates~$\flseg(x, y)$.
Symbols occurring in the spatial conjunctions can also occur in $\Pi$
and~$\Pi'$.

The crux of our method lies in an interaction of the model based
approach to combination of theories~\cite{ModelCobination} and a
so-called $\Match$ function that we propose for establishing logical
implication between a pair of spatial conjunctions.
We use models of $\Pi$, which we call stacks, to guide the process of
showing that every heap that satisfies $\Sigma$ also
satisfies~$\Sigma'$.
In return, the match function collects an assertion that describes
a set of stacks for which the current derivation is also applicable.
This assertion is then used to take those stacks into account for
which we have not proved the entailment yet.
As a result, our method can benefit from the efficiency offered by SMT
for maintaining a logical context keeping track of stacks for which
the entailment is already proved.

In summary, we present (to the best of our knowledge) the first SMT
based decision procedure for separation logic with list segments.
Our main contribution is the entailment checking algorithm for separation logic combined with decidable theories, together with its correctness proof.
Furthermore we provide an implementation of the algorithm using Z3 for theory reasoning, and an evaluation on micro-benchmarks.

The paper is organised as follows. 
A run of the algorithm is illustrated in Section~\ref{sec-illustration}.
We give preliminary definitions in Section~\ref{sec-prelims}.
Our method is described in Section~\ref{sec-mini-alg-cns}.
All proofs are presented in Section~\ref{correctness}.
We present an experimental evaluation in Section~\ref{benchmarks}.
Conclusions are finally presented in Section~\ref{conclu}.

\paragraph{\bfseries Related work}

Our method is directly inspired by a theorem prover for separation
logic with list segments~\cite{SlpPLDI11} based on
paramodulation techniques~\cite{NieRub01} to deal with
equality reasoning.
An approach that turned out quite advantageous compared
to SmallFoot-based proof systems previously developed.

While \cite{SlpPLDI11} only deals with equalities, the work in this paper supports arbitrary SMT theory expressions in the entailment.
Theory extensions of paramodulation are
still an open problem---even state-of-the-art first order
provers deliver poor performance on problems with linear
arithmetic---so it is not evident how to extend \cite{SlpPLDI11} with
theory reasoning.
Similarly, it is unclear how to extend SmallFoot or jStar to obtain a
decision procedure with rich theory reasoning.

Our $\Match$ function can be seen as a generalisation of the unfolding
inferences, geared towards interaction with the logical context
of an SMT solver, rather than literals in a clausal
representation of the entailment problem.
Last but not least, on that previous work the combination with
paramodulation is given by a quite complex inference system, at a
level of detail which would not accessible through a black-box SMT
prover.
The original proof system for list segments~\cite{BerCal04,BerCal05}
gives a starting point to the design of our $\Match$ function.
However, while the proof system needs to branch and perform case
reasoning during proof search, the $\Match$ function is a deterministic, linear pass over the spatial conjuncts.

Recently, entailment between separation logic formulas where $\Pi$
and $\Pi'$ are conjunctions of (dis-)equalities was shown to be
decidable in polynomial time~\cite{Tractable11}.
While we are primarily interested reasoning about rich theory
assertions describing stacks, exploration of this polynomial time
result is an interesting direction for future work.
Regarding an Nelson-Oppen combination of decision
procedures~\cite{NelsonOppen79}, we see an algorithm following this
combination approach as an interesting and difficult question for the
future work.
A direct application of such theory combination does not work, since
it requires a satisfiability checker for sets of
(possibly negated) spatial conjunctions. 
The interplay of conjunction, negation and spatial conjunction is
likely to turn this into a PSPACE problem. 
In contrast, the spatial reasoning in our approach has linear
complexity, thus shifting the computational complexity to the SMT
prover instead.

Chin et al.~\cite{ChinUnfoldingSL12} present a fold/unfold
mechanism to deal with user-specified well-founded recursive
predicates. 
Due to such a general setting, it does not provide completeness. 
Our logic is more restrictive, allowing to develop a complete
decision procedure. 
Similarly, Botin\u{c}an et al.~\cite{Botincan09} rely on a SmallFoot based proof system which, although does not guarantee completeness on the fragment
we consider, is able to deal with user provided inference and
rewriting rules.

}
{\section{Illustration}
\label{sec-illustration}

In this section we illustrate our algorithm using a high-level
description and a simple example. 
To this end we prove the validity of the entailment:
\begin{equation*}
  \underbrace{c < e}_\Pi \land 
  \underbrace{\flseg(a, b) \ast \flseg(a, c) \ast \fnext(c, d) \ast
    \flseg(d, e)}_\Sigma 
  \lthen  
  \underbrace{\ftrue}_{\Pi'}\land
  \underbrace{\flseg(b, c) \ast \flseg(c, e)}_{\Sigma'}\ .
\end{equation*}

Abstractly, the algorithm performs the following key steps.
It symbolically enumerates models that satisfy $\Pi$ and yield a satisfiable heap part for $\Sigma$ in the antecedent.
For each such assignment $s$ the algorithm attempts to (symbolically)
prove that each heap $h$ satisfying the antecedent, i.e., $s, h
\models \Pi\land\Sigma$
also satisfies the consequence, i.e., $s, h\models \Pi'\land
\Sigma'$.
Finally, we generalise the assignment $s$ and use the corresponding
assertion to prune further models of $\Pi$ that would lead to similar
reasoning steps as~$s$.
The entailment is valid if and only all models of the pure parts are
successfully considered.

For our example we begin with the construction of the
constraint that guarantees the satisfiability of the heap part of the
antecedent.
This constraint requires that each pair of spatial predicates in $\Sigma$ is not colliding, i.e., if two predicates start from
the same heap location then one of them represents an empty heap.
A list segment, say $\flseg(a, b)$, represents an empty heap
if its start and end locations are equal, i.e., if $a \feq b$. 
A points-to predicates, say $\fnext(c, d)$, always represents a
non-empty heap.
For the predicates $\flseg(a, b)$ and $\flseg(d, e)$ the absence of
collision is represented as $a \feq d \lthen a \feq b \lor d \feq e$, i.e., if
the start location $a$ of the first predicate is equal to the start
location $d$ of the second predicate then either of the predicates
represents an empty heap.
The remaining pairs of predicates produce the following non-collision
assertions.
\begin{equation*}
  \begin{array}[t]{@{}l@{\qquad\qquad}l@{}}
    a \feq a \lthen a \feq b \lor a \feq c  & 
    \text{$\flseg(a, b)$ and $\fnext(a, c)$}\\[\jot]
    a \feq c \lthen a \feq b \lor \ffalse  & 
    \text{$\flseg(a, b)$ and $\fnext(c, d)$}\\[\jot]
    a \feq d \lthen a \feq b \lor d \feq e &
    \text{$\flseg(a, b)$ and $\flseg(d, e)$}\\[\jot]
    a \feq c \lthen a \feq c \lor \ffalse &
    \text{$\flseg(a, c)$ and $\fnext(c, d)$}\\[\jot]
    a \feq d \lthen a \feq c \lor d \feq e &
    \text{$\flseg(a, c)$ and $\flseg(d, e)$}\\[\jot]
    c \feq d \lthen \ffalse \lor d \feq e & 
    \text{$\fnext(c, d)$ and $\flseg(d, e)$}
  \end{array}
\end{equation*}
We refer to the conjunction of the above assertions as~$\WellFormed(\Sigma)$.

\newpage
Next, we use an SMT solver to find a model for $\Pi \land \WellFormed(\Sigma)$.
If no such model exists the entailment is vacuously true.
For our example, however, the solver finds the model~$s = \set{a \mapsto 0, b \mapsto 0, c \mapsto 0, d \mapsto 1, e \mapsto 1}$.

We then symbolically show that for every heap $h$ model of $\Sigma$ is also a model of~$\Sigma'$. We do this by showing that $\Sigma$ and $\Sigma'$ are
matching, i.e., for each predicate in $\Sigma'$ there is a corresponding `chain' of predicates in $\Sigma$.
The chain condition requires adjacent predicates to
have a location in common, namely, the finish location of
a predicate is equal to the start location of the next with respect to~$s$.

Since matching only needs to deal with predicates representing non-empty
heaps, we first normalise $\Sigma$ and $\Sigma'$ by removing spatial
predicates that are empty in the given model $s$, i.e., we
remove each list segment predicate whose start and finish locations are equal  with respect to~$s$.
From $\Sigma$ we remove $\flseg(a, b)$ since $s(a) = s(b) = 0$, and from
$\Sigma'$ we cannot remove anything.

Now we attempt to find a match for $\flseg(b, c) \in \Sigma'$ in the
normalised antecedent $\flseg(a, c) \ast \fnext(c, d) \ast \flseg(d,
e)$.
The chain should start with $\flseg(a, c)$ since $s(a) =
s(b)$.
Since $\flseg(a, c)$ finishes at the same location as $\flseg(b, c)$
in every model, we are done with the matching for $\flseg(b, c)$.
Since $\flseg(a, c)$ was used to construct a chain, we cannot consider
it in the remaining matching steps (but only for the same model~$s$).
Next we compute matching for $\flseg(c, e) \in \Sigma'$ using the
remaining predicates $\fnext(c, d) \ast \flseg(d, e)$ from~$\Sigma$.
We begin the chain using $\fnext(c, d)$ since it has the same start
location as $\flseg(c, e)$.
Since the finish location of $\fnext(c, d)$ is not equal to $e$
with respect to~$s$ we still need to connect $d$ and~$e$.
We perform this connection by an additional matching request that
requires to match $\flseg(d, e)$ using the remaining predicates from
$\Sigma$, i.e., using only~$\flseg(d, e)$.
Fortunately, this matching request can be trivially satisfied.
Since all predicates of $\Sigma'$ are matched, and all predicates in
$\Sigma$ were used for matching, we conclude that $\Sigma$ and
$\Sigma'$ exactly match with respect to the current~$s$.

The algorithm notices that from the model $s$ only the
assertion $a \feq b$ was necessary to perform the matching.
Hence, the model $s$ is generalised to the assertion $U = (a \feq b)$.
We continue the enumeration of pure models for the antecedent, excluding those where $a \feq b$. 
The SMT solver reports that $\Pi \land \WellFormed(\Sigma) \land \neg U$ is not satisfiable.
Hence we conclude that the entailment is valid.

}
{\section{Preliminaries}\label{sec-prelims}

We write $f\colon X \to Y$ to denote a \emph{function} with domain $X = \dom f$ and \emph{range}~$Y$;
while $f\colon X \pto Y$ is a \emph{partial function} with $\dom f \subseteq X$.
We write $\sepof{f}{n}$ to 
simultaneously denote the union $\unionof{f}{n}$ of $n$ functions, and assert that their domains are pairwise disjoint, i.e. $\dom h_i \cap \dom h_j = \emptyset$ when $i \neq j$.
Given two functions $f\colon Y \to Z$ and $g\colon X \to Y$, we write $f \comp g$ to denote their composition, i.e. $(f \comp g)(x) = f(g(x))$ for every $x \in \dom g$.
%
%
We sometimes write functions explicitly by enumerating their elements,
for example $f = \set{a \mapsto b, b \mapsto c}$ is the function with $\dom f = \set{a, b}$ and such that $f(a) = b$ and $f(b) = c$.

\paragraph{\bfseries Syntax of separation logic}


We assume a sorted language with both theory and uninterpreted symbols.
Each function symbol $f$ has an arity~$n$ and a signature $f \colon \prodof{\tau}{n} \to \tau$, taking $n$ arguments of respective sorts~$\tau_i$ and returning an expression of sort~$\tau$.
A constant symbol is a $0$-ary function symbol.
A \emph{variable} is an uninterpreted constant symbol, and $\Var$ denotes the set of all variables in the language.
Constant and function symbols are combined as usual, respecting their sorts, to build syntactically valid \emph{expressions}. We use $x \colon \tau$ to denote an expression $x$ of sort $\tau$, and $\Sig$ to denote the set of all expressions in the language.

We assume that, among the available sorts, there are $\Int$ and $\Bool$ for, respectively, integer and boolean expressions. We refer to a function symbol of boolean range as a \emph{predicate symbol}, and a boolean expression as a \emph{formula}.
We also assume the existence of a built-in predicate $\feq \colon \tau \times \tau \to \Bool$ for testing equality between two expressions of the same sort; as well as standard theory symbols from the boolean domain, that is: conjunction ($\land$), disjunction~($\lor$), negation~($\lnot$), truth ($\top$), falsity ($\bot$), implication ($\lthen$), bi-implication ($\liff$) and first order quantifiers ($\forall$,~$\exists$).
Theory symbols for arithmetic may also be present, and we use $\fnil$ as an alias for the integer constant $0$.

Additionally, we also define \emph{spatial symbols} to build expressions that describe properties about memory heaps.
We have the spatial predicate symbols $\femp\colon \Bool$, $\fnext\colon \Int \times \Int \to \Bool$ and $\flseg\colon \Int \times \Int \to \Bool$ for, respectively, the empty heap, a points to relation, and acyclic-list segments; their semantics are described in the following section.
Furthermore, we also have the symbol for \emph{spatial conjunction} $\ast\colon\Bool\times\Bool\to\Bool$.
A formula or an expression is said to be \emph{pure} if it contains no spatial symbols.

Although in principle one can write spatial conjunctions of arbitrary boolean formulas, in our context we only deal with the case where each conjunct is a spatial predicate. So when we say a ``spatial conjunction'' what we actually mean is a ``spatial conjunction of spatial predicates''.
Furthermore, at the meta-level, we treat a spatial conjunction $\Sigma = \sepof{S}{n}$ as a multi-set of boolean spatial predicates, and write $\abs{\Sigma} = n$ to denote the number of predicates in the conjunction.
In particular we use set theory symbols to describe relations between spatial predicates and spatial conjunctions, which are always to be interpreted as \emph{multi-set} operations. For example:
\begin{gather*}
\fnext(y, z) \in \flseg(x,y) \ast \fnext(y, z) \\
\fnext(x, y) \ast \fnext(x, y) \not\subseteq \fnext(x, y) \\
\femp \ast \femp \ast \femp \setminus \femp = \femp \ast \femp\;.
\end{gather*}



\paragraph{\bfseries Semantics of separation logic}

Each sort $\tau$ is associated with a set of values, which we also denote by $\tau$, usually according to their background theories; e.g. $\Int = \set{\dots, -1, 0, 1, \dots}$, and $\Bool = \set{\ffalse, \ftrue}$.
We use $\Val = \disjunionof{\tau}{n}$ to denote the disjoint union of all values for all sorts in the language.

A \emph{stack} is a function $s\colon\Var\to\Val$ mapping variables to values in their respective sorts, i.e. for a variable $v\colon \tau$ we have $s(v) \in \tau$.
The domain of $s$ is naturally extender over arbitrary pure expressions in~$\Sig$ using an appropriate interpretation for their theory symbols, e.g. $s(1 + 2) = 3$.
In our context, a \emph{heap} corresponds to a partial function $h\colon\Int\pto\Val$ mapping memory locations, represented as integers, to values.

Given a stack $s$, a heap $h$, and a formula $F$ we inductively define the satisfaction relation of separation logic, denoted $s, h \models F$, as:
\begin{align*}
s, h &\models \Pi           && \text{if $\Pi$ is pure and $s(\Pi) = \ftrue$,} \\
s, h &\models \femp         && \text{if $h = \emptyset$,} \\
s, h &\models \fnext(x, y)  && \text{if $h = \set{s(x) \mapsto s(y)}$,} \\
s, h &\models F_1 \ast F_2  && \text{if $h = h_1 \ast h_2$ for some $h_1$ and $h_2$} \\
&&& \qquad\text{such that $s, h_1 \models F_1$ and $s, h_2 \models F_2$.}
\end{align*}

Semantics for the acyclic list segment is introduced through the inductive definition
$\flseg(x, z) \equiv (x \feq z \land \femp)
  \lor (x \fneq z \land \exists y \dt \fnext(x,y) \ast \flseg(y,z))$.
As an example consider
$\set{ x \mapsto 1, y \mapsto 2}, \set{ 1 \mapsto 3, 3 \mapsto 2 } \models \flseg(x, y)$.

When $s, h \models F$ we say that the interpretation $(s, h)$ is a \emph{model} of the formula $F$.
A formula is \emph{satisfiable} if it admits at least one model, and \emph{valid} if it is satisfied by all possible interpretations.
Note, in particular, that an entailment $F \lthen G$ is valid if every model of $F$ is also a model of $G$.
Finally, for a formula $F$ we write $s \models F$ if it is the case that, for every heap $h$, we have that $s, h \models F$ holds.

Note that $\fnil$ is not treated in any special way by this logic. If one wants $\fnil$ to regain its expected behaviour, i.e. \emph{nothing} can be allocated at the $\fnil$ address, it is enough to consider $\fnext(\fnil, 0) \ast F$, where $F$ is an arbitrary formula.

}
{\begin{figure}[t]
\begin{ezcode}
\StartLineNumbers
\Begin |function| $\Prove(\Pi \land \Sigma \lthen \Pi' \land \Sigma')$
  \State $\Gamma \assign \Pi \land \WellFormed(\Sigma)$ \label{hy:wf}
  \Begin |while exists| $s \models \Gamma$ |do| \label{hy:loop}
    \State $U \assign \Match(s, \Sigma, \Sigma, \Sigma')$
    \State |if| $s \not\models \Pi' \land U$ |then return| $\finvalid$
    \State $\Gamma \assign \Gamma \land \lnot (\Pi' \land U)$
  \End*
  \State |return| $\fvalid$
\End*
\bigskip
\Begin |function| $\Match(s, \SigmaI, \Sigma, \Sigma')$
  \Begin |if exists| $S \in \Sigma$ |such that| $s \models \Empty(S)$
    \State |return| $\Empty(S) \land \Match(s, \SigmaI, \Sigma \setminus S, \Sigma')$ \label{hy:empty_l}
  \EndBegin |if exists| $S' \in \Sigma'$ |such that| $s \models \Empty(S')$
    \State |return| $\Empty(S') \land \Match(s, \SigmaI, \Sigma, \Sigma' \setminus S')$ \label{hy:empty_r}
  \EndBegin |if exists| $S \in \Sigma$, $S' \in \Sigma'$ |such that| $s \models \UnfoldCheck(\SigmaI, S, S')$
    \State |return| $\UnfoldCheck(\SigmaI, S, S') \land \Match(s, \SigmaI,\Sigma \setminus S,(\Sigma' \setminus S') \ast \UnfoldStep(S, S'))$ \label{hy:collide}
  \EndBegin |else|
    \State |return| $(\Sigma \equiv \emptyset) \land (\Sigma' \equiv \emptyset)$ \label{hy:base}
  \End*
\End*
\end{ezcode}
\caption{Model driven entailment checker}
\label{fig-model-driven}
\end{figure}
 
}
{\section{Decision procedure for list segments and SMT theories}
\label{sec-mini-alg-cns}

In this section we define and describe the building blocks that, when put together as shown in the $\Prove$ and $\Match$ procedures of Figure~\ref{fig-model-driven}, constitute a decision procedure for entailment checking.
The procedure works for entailments of the form $\Pi \land \Sigma \lthen \Pi' \land \Sigma'$, where both $\Pi$ and $\Pi'$ are pure formulas, with respect to any background theory supported by the SMT solver, and both $\Sigma$ and $\Sigma'$ are spatial conjunctions.

To abstract away the specifics of a spatial predicate $S$, we first define $\Addr(S)$ and $\Empty(S)$---respectively the \emph{address} and the \emph{emptiness condition} of a given spatial predicate---as follows:
\def\dash{\text{---}}
\begin{equation*}\setlength{\arraycolsep}{10pt}
\begin{array}{ccc}
S           & \Addr(S)  & \Empty(S) \\\hline
\femp       & \dash & \ftrue    \\
\fnext(x,y) & x     & \ffalse   \\
\flseg(x,y) & x     & x \feq y  \\
\end{array}
\end{equation*}
Intuitively, if the emptiness condition is true with respect to a stack-model $s$, the portion of the heap-model that corresponds to $S$ \emph{must} be empty.
Alternatively, if the emptiness condition is false with respect to $s$, the value associated with its address \emph{must} occur in the domain of any heap satisfying the spatial predicate.
Formally:
given $s \models \Empty(S)$ for a stack $s$, we have $s, h \models S$ if, and only if, the heap $h = \emptyset$; and if $s, h \models \lnot\Empty(S) \land S$ then, necessarily, $s(\Addr(S)) \in \dom h$.

\paragraph{\bfseries Well-formedness}

Before introducing the $\WellFormed$ condition, occurring at line~\ref{hy:wf} of the algorithm in Figure~\ref{fig-model-driven}, we first define the notion of \emph{collision} between spatial predicates. Given any two spatial predicates $S$ and $S'$, the formula
\begin{equation*}
\Collide(S, S') =  \lnot\Empty(S) \land \lnot\Empty(S') \land \Addr(S) \feq \Addr(S') \;.
\end{equation*}
states that two predicates collide if, with respect to a stack-model, they are both non-empty and share the same address.
This would cause a problem  if both $S$ and $S'$ occur together in a spatial conjunction, since they would assert that the same address is allocated at two disjoint---separated---portions of the heap.

Given a spatial conjunction $\Sigma = \sepof{S}{n}$, the \emph{well-formedness condition} is defined as the pure formula
\begin{equation*}
\WellFormed(\Sigma) = \smashoperator{\bigland_{1 \leq i < j \leq n}} \lnot\Collide(S_i, S_j) \;,
\end{equation*}
stating that no pair of predicates in the spatial conjunction collide.
As an example consider the spatial conjunction
\begin{equation*}
  \Sigma = \underbrace{\fnext(x, y)}_{S_1} \ast \underbrace{\flseg(x, z)}_{S_2} \ast \underbrace{\fnext(w, z)}_{S_3}
\end{equation*}
we obtain
\begin{align*}
\Collide(S_1, S_2) &= (\ftrue \land x \fneq z \land x \feq x) = (x \fneq z) \\
\Collide(S_1, S_3) &= (\ftrue \land \ftrue \land x \feq w) = (x \feq w) \\
\Collide(S_2, S_3) &= (x \fneq z \land \ftrue \land x \feq w) = (x \fneq z \land x \feq w) \\
\WellFormed(\Sigma) &= \lnot (x \fneq z) \land \lnot (x \feq w) \land \lnot (x \fneq z \land x \feq w) = (x \feq z \lor x \fneq w) \;.
\end{align*}
That is, the formula is well-formed only when $x \feq z$, so that the second predicate is empty, and $x \fneq w$, so that the first and third do not collide.
In general, the well-formedness condition is quite important since, as the next theorem states, it characterises the satisfiability of spatial conjunctions.

\begin{theorem}\label{thm:wff}
A spatial conjunction $\Sigma$ is satisfiable if, and only if, the pure formula $\WellFormed(\Sigma)$ is satisfiable.
\end{theorem}

\paragraph{\bfseries Matching step}

We now proceed towards the introduction of the $\UnfGuard$ condition, used at 
line~\ref{hy:collide} in Figure~\ref{fig-model-driven}, which lies at the core of our matching procedure.
For this we first define, given a spatial conjunction $\Sigma = \sepof{S}{n}$ and an expression $x$, the \emph{allocation condition}
\begin{equation*}
\Alloc(\Sigma,x) = \smashoperator{\biglor_{1 \leq i \leq n}} \lnot\Empty(S_i) \land x \feq \Addr(S_i)
\end{equation*}
which holds, with respect to a stack-model $s$, when a corresponding heap-model~$h$ for $\Sigma$ would necessarily have to include $s(x)$ in its domain.
Continuing from our previous example we have that
\begin{align*}
\Alloc(\Sigma,z) = (\ftrue \land z \feq x) 
                 \lor (x \fneq z \land z \feq x) 
                 \lor (\ftrue \land z \feq w) 
                 = (z \feq x \lor z \feq w) \;.
\end{align*}
That is, the value of $z$ must be allocated in the heap if either $z \feq x$, so it is needed to satisfy $\fnext(x,y)$, or $z \feq w$ and it is needed to satisfy $\fnext(w,z)$. If otherwise the allocation condition is false, although it may occur, there is no actual need for $z$ to be allocated in the domain of the heap.

Now, when trying to prove an entailment $s \models \Sigma \lthen \Sigma'$, we want to show that any heap model of $\Sigma$ is also a model of $\Sigma'$.
Thus, if we find a pair of colliding predicates $S \in \Sigma$ and $S' \in \Sigma'$, then portion of the heap that satisfies $S$ should overlap with the portion of the heap that satisfies $S'$.
In fact, it is not hard to convince oneself---for the list segment predicates considered---that the heap model of $S'$ should match exactly that of $S$ plus some extra surplus.

In the following definitions $\UnfUpdate$ gives the precise value of the extra surplus, while $\UnfCheck$ specifies additional conditions which are necessary so that the model of $S$ doesn't leak outside the model of $S'$.
\begin{equation*}\setlength{\arraycolsep}{10pt}
\begin{array}{cc|cc}
S' & S & \UnfUpdate(S,S') & \UnfCheck(\Sigma,S,S') \\\hline
\fnext(x',z) & \fnext(x,y) & \femp       & y \feq z  \\
\flseg(x',z) & \fnext(x,y) & \flseg(y,z) & \ftrue    \\
\fnext(x',z) & \flseg(x,y) & \femp       & \ffalse   \\
\flseg(x',z) & \flseg(x,y) & \flseg(y,z) & y \fneq z \lthen \Alloc(\Sigma,z)
\end{array}
\end{equation*}
The \emph{matching step condition} is the formula
\begin{equation*}
\UnfGuard(\Sigma, S, S') = \Collide(S, S') \land \UnfCheck(\Sigma, S, S') \;.
\end{equation*}

To formalise our stated intuition, the following proposition articulates how the residue that is computed between two colliding predicates is indeed satisfied by the remaining heap surplus.
The validity of this statement, as in the case of the subsequent two propositions, can be easily verified by inspection of the relevant definitions.

\begin{proposition}\label{prop:step_unfold}
Given two spatial predicates $S$, $S'$, a stack $s \models \Collide(S,S')$ and a heap $h$ such that $s, h \models S'$, if there is a partition $h = h_1 \ast h_2$ for which $s, h_1 \models S$, it necessarily follows that $s, h_2 \models \UnfUpdate(S, S')$.
\end{proposition}

Moreover, for any stack satisfying the matching step condition, we are free to replace $S'$ in $\Sigma'$ with the matched expression $S \ast \UnfUpdate(S, S')$. Formally we state the following proposition.

\begin{proposition}\label{prop:step_fold}
Given a stack $s \models \UnfGuard(\Sigma, S, S')$, where $S$ and $S'$ are spatial predicates, and $S$ occurs in the spatial conjunction $\Sigma$, for any spatial conjunction $\Sigma'$ containing $S'$ we have that
\begin{equation*}
s \models (\Sigma' \setminus S') \ast S \ast \UnfUpdate(S, S') \lthen \Sigma'
\end{equation*}
\end{proposition}

Finally, we state that the enclosing condition is complete in the sense that, if it were not satisfied by a stack $s$, then one could build a counterexample for the matching $S \ast \UnfUpdate(S,S') \lthen S'$.

\begin{proposition}\label{prop:step_twist}
Given two spatial predicates $S$, $S'$, a spatial conjunction $\Sigma$ that contains $S$, a stack $s$ and a two-part heap $h = h_1 \ast h_2$ such that $s, h_1 \ast h_2 \models \Sigma$ and $s, h_2 \models S \ast \UnfUpdate(S,S')$,
if $s \models \Collide(S, S') \land \lnot\UnfCheck(\Sigma, S, S')$, then there is a $h_2'$ such that $s, h_1 \ast h'_2 \models \Sigma$ but $s, h'_2 \not\models S \ast \UnfUpdate(S,S') \lthen S'$.
\end{proposition}

As an example consider the case where $S = \flseg(x,y)$ and $S' = \flseg(x',z)$, such that $\UnfUpdate(S,S') = \flseg(y,z)$.
Take some stack $s \models \Collide(S,S')$ and the heap $h_2 = \set{s(x) \mapsto s(y), s(y) \mapsto s(z)}$ as a model of $\flseg(x,y) \ast \flseg(y,z)$.
From $s \models \lnot\UnfCheck(\Sigma, S, S')$ it follows that $s(x) \neq s(y)$ and the address $s(z)$ does not need to be allocated anywhere in $h = h_1 \ast h_2$.
This allows us to patch and let $h'_2 = \set{s(x) \mapsto s(z), s(z) \mapsto s(y), s(y) \mapsto s(z)}$, which is still a model of the pair $\flseg(x,y) \ast \flseg(y,z)$ but---due to the introduced cycle---not of $\flseg(x',z)$.

\paragraph{\bfseries Matching and proving}

To finalise the description of our decision procedure for entailment checking we have only left to put all the ingredients together, as shown in Figure~\ref{fig-model-driven}, into the $\Match$ and $\Prove$ functions.

The $\Match$ function tries to establish whether $s \models \SigmaI \lthen (\SigmaI \setminus \Sigma) \ast \Sigma'$.
Initially called with $\SigmaI$ set to $\Sigma$, at the top level this is in fact equivalent to checking the validity of $s \models \Sigma \lthen \Sigma'$.
During the execution process $\SigmaI$ will retain its initial value, $\Sigma$ and $\Sigma'$ carry the portions of the entailment that are left to match, while $\SigmaI \setminus \Sigma$ is the fragment already matched.
As the function progresses, the conjunctions $\Sigma$ and $\Sigma'$ will become shorter, while the matched portion $\SigmaI \setminus \Sigma$ grows. 
If successful both $\Sigma$ and $\Sigma'$ will become empty, yielding at the end the trivial entailment $s \models \SigmaI \lthen \SigmaI$.

The function begins by inspecting $\Sigma$ and $\Sigma'$ to discard, at lines \ref{hy:empty_l} and \ref{hy:empty_r}, any empty predicates with respect to $s$, and recursively calling itself to verify the rest of the entailment.
After removing all such empty predicates, if a valid matching step is found, the predicate $S'$ occurring in $\Sigma'$ is replaced with $S \ast \UnfUpdate(S, S')$, so that $S$---which now occurs both in $\Sigma$ and $\Sigma$'---can be moved to the matched part of the entailment in the recursive call at line~\ref{hy:collide}.

If the function is successful, after reaching the bottom of the recursion at line~\ref{hy:base} with both $\Sigma$ and $\Sigma'$ becoming empty, the return value collects a conjunction of all assumptions made on the values of stack.
This allows to generalise the proof which works not only for the particular stack $s$, but for any stack satisfying the same assumptions.
Otherwise, if the bottom of the recursion is reached with some portions still left to match, the function returns an unsatisfiable formula signalling the existence of a counterexample for the entailment.
This behaviour is formalised in the following theorem, proved later in Section~\ref{correctness}.

\begin{theorem}\label{thm:match}
Given a pair of spatial conjunctions $\Sigma$, $\Sigma'$ and a stack $s$ such that $s \models \WellFormed(\SigmaI)$, we have that:
\begin{itemize}
\item the procedure $\Match(s, \Sigma, \Sigma, \Sigma')$ always
  terminates with a result $U$,
\item the execution requires $O(n)$ recursive steps, where $n = \abs{\Sigma} + \abs{\Sigma'}$.
\item if $s \models U$ then the entailment $U \land \Sigma \lthen \Sigma'$ is valid, and
\item if $s \not\models U$ then $s \not\models \Sigma \lthen \Sigma'$.
\end{itemize}
\end{theorem}

The main $\Prove$ function, which checks whether $\Pi \land \Sigma \lthen \Pi' \land \Sigma'$ is valid, begins with the pure formula $\Gamma \assign \Pi \land \WellFormed(\Sigma)$.
An SMT solver iteratively finds models for $\Gamma$, which become candidate stack models to guide the search for a proof or a counterexample.
Given one such stack $s$, the $\Match$ function is called to check the validity of the entailment with respect to~$s$.
If successful, $\Match$ returns a formula $U$ generalising the conditions in which the entailment is valid, so the search may continue for stacks where $U$ does not hold.
The iterations proceed until either all possible stacks have been discarded, or a counterexample is found in the process.
It is important to stress that the function does not enumerate all concrete models but, rather, the equivalence classes returned by $\Match$.
Formally we state the following theorem, whose proof is given in Section~\ref{correctness}.

\begin{theorem}\label{thm:prove}
Given two pure formulas $\Pi$, $\Pi'$, and two spatial formulas $\Sigma$, $\Sigma'$, we have that:
\begin{itemize}
\item the procedure $\Prove(\Pi \land \Sigma \lthen \Pi' \land \Sigma')$ always terminates, and
\item the return value corresponds to the validity of $\Pi \land \Sigma \lthen \Pi' \land \Sigma'$.
\end{itemize}
\end{theorem}

}
{\section{Proofs of correctness}
\label{correctness}

This section presents the main technical contribution of the paper, the proof of correctness of our entailment checking algorithm.
The proof itself closely follows the structure of the previous section, filling in the technical details required to assert the statements of Theorem~\ref{thm:wff}, on well-formedness, Theorem~\ref{thm:match}, on matching, and finally Theorem~\ref{thm:prove} on entailment checking.

\paragraph{\bfseries Well-formedness}

Soundness of the well-formed condition $\WellFormed(\Sigma)$, the first half of Theorem~\ref{thm:wff}, can be easily shown by noting that if a spatial conjunction $\Sigma$ is satisfiable with respect to some stack and a heap, the formula $\WellFormed(\Sigma)$ is also necessarily true with respect to the same stack.

\begin{proposition}\label{prop:sound_wf}
Given a spatial conjunction $\Sigma$, a stack $s$, and a heap $h$, if we have $s, h \models \Sigma$, then also $s \models \WellFormed(\Sigma)$.
\end{proposition}

\begin{proof}
Let $\Sigma  = \sepof{S}{n}$. Since $s, h \models \Sigma$, there is a partition $h = \sepof{h}{n}$ such that each $s, h_i \models S_i$.
Given a pair of predicates $S_i$ and $S_j$ with $i < j$, if either $s \models \Empty(S_i)$ or $s \models \Empty(S_j)$, then trivially $s \models \lnot\Collide(S_i, S_j)$.

Assume otherwise that $s \models \lnot\Empty(S_i) \land \lnot\Empty(S_j)$.
It follows that both $s(\Addr(S_i)) \in \dom h_i$ and $s(\Addr(S_j)) \in \dom
h_j$. Since by construction  $h_i$ and~$h_j$ have disjoint domains, we have $s(\Addr(S_i)) \neq s(\Addr(S_j))$. This implies the fact that $s \models \lnot\Collide(S_i, S_j)$. \qed
\end{proof}

For completeness of the well-formed condition $\WellFormed(\Sigma)$, the second half of Theorem~\ref{thm:wff}, we prove a slightly more general result. In particular we show that if a stack $s \models \WellFormed(\Sigma)$ then it is possible to build a heap $h$ such that $s, h \models \Sigma$. Furthermore, we show that such $h$ is \emph{conservative} in the sense that it only allocates addresses which are strictly necessary.

\begin{proposition}
Given a spatial conjunction $\Sigma = \sepof{S}{n}$ and a stack~$s$ such that $s \models \WellFormed(\Sigma)$, there is a heap $h$ for which $s, h \models \Sigma$ and, furthermore, the domain $\dom h = \set{ \Addr(S_i) \mid s \models \lnot\Empty(S_i) }$.
\end{proposition}

\begin{proof}\label{prop:comp_wf}
Consider the heap $h = \sepof{h}{n}$ where each $h_i$ is defined as follows:
\begin{itemize}
\item if $s \models \Empty(S_i)$ then $h_i = \emptyset$; otherwise
\item if $s \models \lnot\Empty(S_i)$ it follows that $S_i = \fnext(x,y)$ or $S_i = \flseg(x,y)$, in either case let $h_i = \set{ s(x) \mapsto s(y) }$.
\end{itemize}
By construction $s, h_i \models S_i$ and, furthermore, if $s \models \lnot\Empty(S_i,S_j)$ it follows that $\dom h_i = \set{s(\Addr(S_i))}$.
From this we easily get as desired that the domain of the heap $\dom h = \set{ \Addr(S_i) \mid s \models \lnot\Empty(S_i) }$.
Now, to prove that $s, h \models \Sigma$, we have only left to show that for any pair $S_i$, $S_j$ with $i \neq j$ the domains or their respective heaplets are disjoint, i.e. $\dom h_i \cap \dom h_j = \emptyset$.

If either $s \models \Empty(S_i)$ or $s \models \Empty(S_j)$ the result is trivial. Otherwise assume that $s \models \lnot\Empty(S_i) \land \lnot\Empty(S_j)$. Since $s \models \WellFormed(\Sigma)$, and in particular also $s \models \lnot\Collide(S_i, S_j)$, it follows that $s \not\models \Addr(S_i) \feq \Addr(S_j)$.
Namely the address values $s(\Addr(S_i)) \neq s(\Addr(S_j))$ and, thus, the domains of $h_i$ and $h_j$ are disjoint. \qed
\end{proof}

Theorem~\ref{thm:wff} follows immediately as a corollary of Propositions~\ref{prop:sound_wf} and~\ref{prop:comp_wf}.

\paragraph{\bfseries Matching and proving}

The following proposition is the main ingredient required to establish the soundness and completeness of the $\Match$ procedure of Figure~\ref{fig-model-driven}.
The proof, although long and quite technical in details, follows the intuitive description given in Section~\ref{sec-mini-alg-cns} about the behaviour of $\Match$.
Each of the main four cases in the proof corresponds, respectively to the conditions on lines~\ref{hy:empty_l} and~\ref{hy:empty_r}, when discarding empty predicates, line~\ref{hy:collide}, when a matching step is performed, and finally line~\ref{hy:base}, when the base case of the recursion is reached.

Each case is further divided in two sub-cases, one for the situation when the recursive call is successful and a proof of validity is established, and one for the situation when a counterexample is built.
The last case, the base of the recursion, is divided into four sub-cases: the successful case when the matching is completed, the case in which all of $\Sigma'$ is consumed but there are predicates in $\Sigma$ left to match, the case in which there is a collision but the enclosure condition is not met, and finally the case in which there is no collision at all.


\begin{proposition}
Given three spatial formulas $\SigmaI$, $\Sigma$, $\Sigma'$, and a stack $s$ such that $\Sigma \subseteq \SigmaI$, and $s \models \WellFormed(\SigmaI)$; let~$U$ be the pure formula returned by $\Match(s, \SigmaI, \Sigma, \Sigma')$.
\begin{itemize}
\item If $s \models U$ then $U \land \SigmaI \lthen (\SigmaI \setminus \Sigma) \ast \Sigma'$ is valid and, otherwise
\item if $s \not\models U$ there is a $h$ such that $s, h \not\models \SigmaI \lthen (\SigmaI \setminus \Sigma) \ast \Sigma'$.
\end{itemize}
\end{proposition}


\begin{proof}
The proof goes by induction, following the recursive definition of the $\Match$ function.
\begin{itemize}
\item Suppose we reach line~\ref{hy:empty_l}, with a predicate $S \in \Sigma$ such that $s \models \Empty(S)$.
Recursively let $U' = \Match(s, \SigmaI, \Sigma \setminus S, \Sigma')$ and $U = \Empty(S) \land U'$.
Since $s \models \Empty(S)$ it follows $s \models U \liff U'$.
\begin{itemize}
\item if $s \models U$, we want to show that $U \land \SigmaI \lthen (\SigmaI \setminus \Sigma) \ast \Sigma'$ is valid, so take any model $s', h\models U \land \SigmaI$.
By induction we know the formula $U \land \SigmaI \lthen R$ is valid, where
$R = (\SigmaI \setminus (\Sigma \setminus S)) \ast \Sigma'
  = (\SigmaI \setminus \Sigma) \ast S \ast \Sigma'$.
It follows therefore follows that $s',h \models (\SigmaI \setminus \Sigma) \ast S \ast \Sigma'$.
Since $s \models \Empty(S)$, there is nothing allocated in $h$ for $S$ and, thus, $s',h \models (\SigmaI \setminus \Sigma) \ast \Sigma'$.
\item if $s \not\models U'$, by induction there is a heap $h$ such that $s, h \models \SigmaI$ but, at the same time, $s, h \not\models (\SigmaI \setminus \Sigma) \ast S \ast \Sigma'$.
Again, since $s, \emptyset \models S$, it must be the case that $s, h \not\models (\SigmaI \setminus \Sigma) \ast \Sigma'$. (Otherwise you get a contradiction.)
\end{itemize}
\item Suppose we reach line~\ref{hy:empty_r} with a predicate $S' \in \Sigma$ such that $s \models \Empty(S')$.
Recursively let $U' = \Match(s, \SigmaI, \Sigma, \Sigma' \setminus S')$ and $U = \Empty(S') \land U'$.
Again we have $s \models U \liff U'$.
\begin{itemize}
\item if $s \models U'$, we want to show that $U' \land \SigmaI \lthen (\SigmaI \setminus \Sigma) \ast \Sigma'$ is valid, so take any model $s', h\models U \land \SigmaI$.
By induction we know $U \land \SigmaI \lthen (\SigmaI \setminus \Sigma) \ast (\Sigma' \setminus S')$ is valid and, thus, we also get that $s', h \models (\SigmaI \setminus \Sigma) \ast (\Sigma' \setminus S')$.
Again, from $s' \models \Empty(S')$ and $s', \emptyset \models S'$ it follows
$s', h \models (\SigmaI \setminus \Sigma) \ast (\Sigma' \setminus S') \ast S'$
or, equivalently, $s', h \models (\SigmaI \setminus \Sigma) \ast \Sigma'$.
\item if $s \not\models U'$, by induction there is a heap $h$ such that $s, h \models \SigmaI$ but, at the same time, $s, h \not\models (\SigmaI \setminus \Sigma) \ast (\Sigma' \setminus S')$.
Similarly $s, \emptyset \models S'$, so it must be the case that $s, h \not\models (\SigmaI \setminus \Sigma) \ast (\Sigma' \setminus S') \ast S'$ or, equivalently, $s, h \not\models (\SigmaI \setminus \Sigma) \ast \Sigma'$.
\end{itemize}
\item Suppose we reach line~\ref{hy:collide}, with two of predicates $S \in \Sigma$ and $S' \in \Sigma'$, such that the stack $s \models \UnfGuard(\SigmaI, S, S')$.
Let $S'' = \UnfUpdate(S, S')$, recursively obtain $U' = \Match(s, \SigmaI, \Sigma \setminus S, (\Sigma' \setminus S') \ast S'')$ and let $U = \UnfGuard(S) \land U'$.
As before we have $s \models U \liff U'$.
\begin{itemize}
\item if $s \models U$, we want to show that $U \land \SigmaI \lthen (\SigmaI \setminus \Sigma) \ast \Sigma'$ is valid. That is, any model $s', h \models U \land \SigmaI$ is also a model of $(\SigmaI \setminus \Sigma) \ast \Sigma'$.
By induction we have that $U' \land \SigmaI \lthen R$ is valid, where the formula
\begin{equation*}
R = (\SigmaI \setminus (\Sigma \setminus S)) \ast (\Sigma' \setminus S') \ast S'' 
  = (\SigmaI \setminus \Sigma) \ast (\Sigma' \setminus S') \ast S \ast S''\;.
\end{equation*}
Since $s',h \models U' \land \hat\Sigma$ it follows that $s',h \models R$.
By Proposition~\ref{prop:step_fold}, since $s' \models \UnfGuard(\SigmaI,S,S')$, we obtain that
$s', h \models (\SigmaI \setminus \Sigma) * (\Sigma' \setminus S') \ast S'$
or, equivalently, $s', h \models (\SigmaI \setminus \Sigma) \ast \Sigma'$.
\item if $s \not\models U$, by induction, there exists a heap $h$ such that $s,h \models \hat\Sigma$ but, however,
$s,h \not\models (\SigmaI \setminus \Sigma) \ast (\Sigma' \setminus S') \ast S \ast S''$.
Partition $h = h_1 \ast h_2$ such that $s, h_1 \models \SigmaI \setminus S$ and $s, h_2 \models S$.
Now note that, regardless of the value of $S$, letting $h'_2 = \set{s(x) \mapsto s(y)}$ and $h' = h_1 \ast h'_2$ we have that both $s, h'_2 \models S$ and $s, h' \models \hat\Sigma$. We claim that $s, h' \not\models (\SigmaI, \setminus \Sigma) \ast \Sigma'$.

Assume by contradiction that $s, h' \models (\SigmaI \setminus \Sigma) \ast \Sigma'$, and partition now $h' = h_3 \ast h_4$ such that $s, h_3 \models (\SigmaI \setminus \Sigma) \ast (\Sigma' \setminus S')$ and $s, h_4 \models S'$.
Because $S$ and $S'$ collide, it follows that $\dom h'_2 = \set{s(\Addr(S))} \subseteq \dom h_4$ and $h_4 = h'_2 \ast h_5$ for some remainder~$h_5$.
Then, by Proposition~\ref{prop:step_unfold}, $s, h_4 \models S \ast S''$ and $s, h_5 \models S''$.
But $h = h_1 \ast h_2 = h_3 \ast h_2 \ast h_5$ would make a model of $(\SigmaI \ast \Sigma) \ast (\Sigma' \setminus S') \ast S \ast S''$, contradicting our inductive hypothesis.
\end{itemize}
\item Suppose we reach line~\ref{hy:base}. We can find ourselves in several situations:
\begin{itemize}
\item $\Sigma' = \emptyset$, $\Sigma = \emptyset$, and the function returns $U = \ftrue$.
In this case it is trivial that $s \models U$ and $U \land \SigmaI \lthen (\SigmaI \setminus \emptyset) \ast \emptyset$ is valid.
\item $\Sigma' = \emptyset$, there is a $S \in \Sigma$, and the function returns $U = \ffalse$.
In this case $s \not\models U$, so we need to find a counterexample for the entailment.
From Proposition~\ref{prop:comp_wf} there is a heap $h$ such that $s, h \models \SigmaI$.
Partition $h = h_1 \ast h_2$ such that $s, h_1 \models (\SigmaI \setminus \Sigma)$ and $s, h_2 \models \Sigma$.
Since~$S$ occurs in $\Sigma$, and at this point $s \not\models \Empty(S)$, it is necessarily the case that $s(\Addr(S)) \in \dom h_2$. In particular $h_2 \neq \emptyset$, and because $h = h_1 \ast h_2$, we obtain $s, h \not\models (\SigmaI \setminus \Sigma)$.
Furthermore, since $\Sigma' = \emptyset$, this is equivalent to $s, h \not\models (\SigmaI \setminus \Sigma) \ast \Sigma'$.
\item There is a $S' \in \Sigma'$, a $S \in \Sigma$ such that $s \models \Collide(S, S')$, and the function returns $U = \ffalse$.
Since we did not end up on line~\ref{hy:collide}, it must be the case that $s \not\models \UnfCheck(\SigmaI, S, S')$.
By Property~\ref{prop:comp_wf} there is a heap $h$ such that $s, h \models \SigmaI$.
Partition $h = h_1 \ast h_2$ such that $s, h_1 \models (\SigmaI \setminus S)$ and $s, h_2 \models S$.
Let $h'_2 = \set{s(x) \mapsto s(y)}$ and $h' = h_1 \ast h'_2$; since $s \models \lnot\Empty(S)$ we have that $s, h'_2 \models S$ and $s, h' \models \SigmaI$.

If it turns out that $s, h' \not\models (\SigmaI \setminus \Sigma) \ast \Sigma'$ we are done.
Assume otherwise that $s, h' \models (\SigmaI \setminus \Sigma) \ast \Sigma'$ and partition the heap $h' = h_3 \ast h_4$ such that $s, h_3 \models (\SigmaI \setminus \Sigma) \ast (\Sigma' \setminus S')$ and $s, h_4 \models S'$.
Since the predicates $S$, $S'$ collide and are non-empty, it follows that the address $\set{s(\Addr(S))} = \dom h'_2 \subseteq \dom h_4$ and, therefore, $h_4 = h'_2 \ast h_5$ for some remainder~$h_5$.
By Proposition~\ref{prop:step_unfold} it follows that $s, h_4 \models S \ast S''$ and $s, h_5 \models S''$.
Since $h' = h_3 \ast h'_2 \ast h_5$ it follows then that $s, h_3 \ast h_5 \models (\SigmaI \setminus S)$.
By Proposition~\ref{prop:step_twist} there is a $h_6$ such that $s, h_3 \ast h_5 \ast h_6 \models \SigmaI$ but $s, h_5 \ast h_6 \not\models S'$.
However, since $s, h_3 \models (\SigmaI \setminus \Sigma) \ast (\Sigma' \setminus S')$, it follows that $s, h_3 \ast h_5 \ast h_6 \not\models (\SigmaI \setminus \Sigma) \ast \Sigma'$.
The heap $h_3 \ast h_5 \ast h_6$ is a counterexample for the entailment.
\item There is some $S' \in \Sigma'$ and $s \not\models \Collide(S, S')$ for all $S \in \Sigma$, thus the function returns $U = \ffalse$.
By Property~\ref{prop:comp_wf} there is a heap $h$ such that $s, h \models \SigmaI$.
Partition $h = h_1 \ast h_2$ into two parts such that $s, h_1 \models (\SigmaI \setminus \Sigma)$ and $s, h_2 \models \Sigma$.
Since $S'$ does not collide with any predicate in $\Sigma$, it follows that $s(\Addr(S')) \notin \dom h_2$, in particular $s, h_2 \not\models \Sigma'$.
From this it follows that $s, h_1 \ast h_2 \not\models (\SigmaI \setminus \Sigma) \ast \Sigma'$. \qed
\end{itemize}
\end{itemize}
\end{proof}

The correctness of the $\Match$ procedure, formally stated previously in Theorem~\ref{thm:match}, follows as a corollary of this proposition for the case when $\SigmaI = \Sigma$.
Termination of the procedure can also be easily verified since, at the recursive calls in lines~\ref{hy:empty_l} and~\ref{hy:collide} the size of the third argument decreases and, when it stays the same at the recursive call in line~\ref{hy:empty_r}, the size of the fourth argument decreases.
This same termination argument also shows that the number of recursive steps is in fact linear in the size of $\Sigma$ and $\Sigma'$.

Finally we are ready to prove the termination and correctness of the main $\Prove$ procedure as stated earlier in Theorem~\ref{thm:prove}.
Specifically, we'll show that the procedure returns $\fvalid$ if, and only if, the entailment $\Pi \land \Sigma \lthen \Pi' \land \Sigma'$ supplied as argument is indeed valid.

\begin{proof}[of Theorem~\ref{thm:prove}]
Termination can be established since at each iteration of the loop at line~\ref{hy:loop}, the number satisfying models of $\Gamma$ is being strictly reduced. Since there is only a finite number of formulas that can be built by combinations of $\Empty(S)$ and $\UnfGuard(\SigmaI, S, S')$---the building blocks for $U$---all suitable combinations should be exhausted at some point.

For correctness we now prove that line~\ref{hy:loop} at the base of the loop always satisfies the invariants:
\begin{enumerate}
\item $\Gamma \lthen \Pi \land \WellFormed(\Sigma)$, and
\item if $\Gamma \land \Sigma \lthen \Pi' \land \Sigma'$ is valid then also $\Pi \land \Sigma \lthen \Pi' \land \Sigma'$ is.
\end{enumerate}

The first invariant can be easily verified by inspecting the code and noting that at the beginning $\Gamma = \Pi \land \WellFormed(\Sigma)$, and later only more conjuncts are appended to $\Gamma$.  

For the second invariant, right before entering the loop we have that $\Gamma = \Pi \land \WellFormed(\Sigma)$.
So, assuming that $\Pi \land \WellFormed(\Sigma) \land \Sigma \lthen \Pi' \land \Sigma'$ is valid, take any $s', h \models \Pi \land \Sigma$, from Proposition~\ref{prop:sound_wf} it follows that $s' \models \WellFormed(\Sigma)$ and therefore, from our assumption, $s', h \models \Pi' \land \Sigma'$.

If we enter the code of the loop we have that $s \models \Gamma$ and start by letting $U = \Match(s, \Sigma, \Sigma, \Sigma')$. If $s \not\models \Pi' \land U$, then either we have that $s \not\models \Pi'$---from Proposition~\ref{prop:comp_wf} there is a heap $h$ such that $s, h \models \Pi \land \Sigma$ but $s, h \not\models \Pi'$---or $s \not\models U$---in which case from Theorem~\ref{thm:match} there is a $h$ such that $s, h \models \Pi \land \Sigma$ but $s, h \not\models \Sigma'$.
In either case the entailment is invalid and the procedure correctly reports this.

Alternatively, if $s \models \Pi' \land U$, from $\Gamma \land \lnot(\Pi' \land U) \land \Sigma \lthen \Pi' \land \Sigma'$ we have to prove that $\Pi \land \Sigma \lthen \Pi' \land \Sigma'$.
Take any $s', h \models \Pi \land \Sigma$, if $s', h \models \Pi' \land U$ then from Theorem~\ref{thm:match} the formula $U \land \Sigma \lthen \Sigma'$ is valid, and $s', h \models \Pi' \land \Sigma'$. Otherwise, if $s', h \not\models \Pi' \land U$, from our assumption we have as well $s', h \models \Pi' \land \Sigma'$. \qed
\end{proof}

}
{\section{Experiments}\label{benchmarks}

We implemented our entailment checking algorithm in a tool called \SlpMT\ using~Z3 as the theory back-end for testing the satisfaction of pure formulas and evaluating expressions against pure stack-models.
The tool already accepts arbitrary theory expressions and assertions as part of the entailment formula.
However, due to the current lack of realistic application benchmarks making use of such theory features, we only report the running times of this new implementation against already published benchmarks from~\cite{SlpPLDI11}.

\newcommand{\timeout}[1]{{(#1\%)}}

\begin{table}[t]
\centering

\begin{formal-table}{
  number = [1]   {Copies},
    number = [3.2] {\SmallFoot},
    number = [3.2] {\Separation},
    number = [1.2] {\SlpMT}
}
1 & 0.01 & 0.11 & 0.17 \\
2 & 0.07 & 0.06 & 0.19 \\
3 & 1.03 & 0.08 & 0.23 \\
4 & 9.53 & 0.13 & 0.26 \\
5 & 55.85 & 0.38 & 0.31 \\
6 & 245.69 & 2.37 & 0.39 \\
7 & \timeout{64} & 20.83 & 0.54 \\
8 & \timeout{15} & 212.17 & 0.85 \\
9 & {---} & {---} & 1.49 \\
10 & {---} & {---} & 2.81
\end{formal-table}

\caption{Running times in seconds while checking `clones' of \SmallFoot\ examples.} 
\label{tab:clones}
\end{table}

Table~\ref{tab:clones} shows experiments that have a significant number of repeated spatial atoms in the entailment. 
They are particularly difficult for the unfolding implemented in \Separation\ and the match function in~\SlpMT.
Since our match function collects constraints that can potentially be useful for other applications of match, we observe a significant improvement.

}

\section{Conclusion}\label{conclu}

We have presented a method for extending an SMT solver with separation logic using the list segment predicate.
Our method decides entailments of the form $\Pi \land \Sigma
\rightarrow \Pi' \land \Sigma'$, whose pure and spatial components may freely use arbitrary theory assertions and theory expressions, as long as they are supported by the back-end SMT solver.
Furthermore, we provide a formal proof of correctness of the algorithm, as well as a experimental results with an implementation using~Z3 as the theory solver.

\bibliographystyle{abbrv}
\bibliography{biblio,biblio1}

\end{document}